\documentclass[letterpaper, 10pt, conference]{ieeeconf}  

\usepackage{graphicx}
\usepackage{subfig}
\usepackage{amsmath}
\usepackage{amsthm} 
\usepackage{amssymb}
\usepackage[font=footnotesize]{caption} 
\usepackage{subfig} 
\usepackage[noadjust]{cite} 
\usepackage{color}
\usepackage{algorithm} 
\usepackage{algpseudocode} 
\usepackage{enumerate}
\usepackage[hidelinks,colorlinks=false]{hyperref}
\usepackage{setspace}
\usepackage{arydshln} 

\usepackage{tikz}
\usetikzlibrary{calc} 
\usetikzlibrary{shapes} 
\usetikzlibrary{chains}
\usetikzlibrary{fit}
\usetikzlibrary{arrows}
\usetikzlibrary{decorations.text} 
\usetikzlibrary{decorations.markings}

\newtheorem{theorem}{Theorem}
\newtheorem{lemma}{Lemma}

\newtheorem{remark}{Remark}

\newtheorem{corollary}{Corollary}

\newtheorem{definition}{Definition}

\newcommand{\Null}{\mathrm{Null}}
\newcommand{\Range}{\mathrm{Range}}
\newcommand{\one}{\mathbf{1}}
\newcommand{\rank}{\mathrm{rank}}
\newcommand{\myspan}{\mathrm{span}}

\newcommand{\R}{\mathbb{R}}
\newcommand{\G}{\mathcal{G}}
\newcommand{\E}{\mathcal{E}}
\newcommand{\V}{\mathcal{V}}
\newcommand{\N}{\mathcal{N}}

\renewcommand{\L}{\mathcal{B}}

\graphicspath{{figures/}}

\begin{document}

\title{Laman Graphs are Generically Bearing Rigid in Arbitrary Dimensions}
\author{Shiyu Zhao, Zhiyong Sun, Daniel Zelazo, Minh-Hoang Trinh, and Hyo-Sung Ahn 
\thanks{S. Zhao is with the Department of Automatic Control and Systems Engineering, University of Sheffield, United Kingdom.
{\tt\small szhao@sheffield.ac.uk}}
\thanks{Z. Sun is with the Research School of Engineering, Australian National University, Australia.
{\tt\small zhiyong.sun@anu.edu.au}}
\thanks{D. Zelazo is with the Faculty of Aerospace Engineering, Technion--Israel Institute of Technology, Israel.
{\tt\small dzelazo@technion.ac.il}}
\thanks{M.-H. Trinh and H.-S. Ahn are with the Distributed Control \& Autonomous Systems Lab, Gwangju Institute of Science and Technology, Korea.
{\tt\small \{trinhhoangminh, hyosung\}@gist.ac.kr}}
}
\IEEEoverridecommandlockouts
\maketitle
\begin{abstract}
This paper addresses the problem of constructing bearing rigid networks in arbitrary dimensions. We first show that the bearing rigidity of a network is a generic property that is critically determined by the underlying graph of the network. A new notion termed generic bearing rigidity is defined for graphs. If the underlying graph of a network is generically bearing rigid, then the network is bearing rigid for almost all configurations; otherwise, the network is not bearing rigid for any configuration. As a result, the key to construct bearing rigid networks is to construct generically bearing rigid graphs. The main contribution of this paper is to prove that Laman graphs, which can be generated by the Henneberg construction, are generically bearing rigid in arbitrary dimensions. As a consequence, if the underlying graph of a network is Laman, the network is bearing rigid for almost all configurations in arbitrary dimensions.
\end{abstract}
\overrideIEEEmargins

\section{Introduction}

Consider a network defined as a graph with its vertices mapped to a set of distinct points in the Euclidean space. Such a network may represent a sensor network or multi-robot system. The bearing rigidity is a fundamental property of the network that indicates whether the geometric pattern can be uniquely determined by the inter-node bearings. The bearing rigidity theory has received increasing research attention in recent years in the area of multi-robot formation control and sensor network localization because it provides an architectural condition for the control and estimation algorithms to converge \cite{eren2003,bishopconf2011rigid,Eren2012IJC,zhao2014TACBearing,zhao2015NetLocalization}.
The necessary and sufficient conditions for the bearing rigidity of networks in arbitrary dimensions have been proved using the bearing rigidity matrix \cite[Thm~4]{zhao2014TACBearing} and the bearing Laplacian matrix \cite[Lem~2]{zhao2015NetLocalization}, respectively. However, the problem of \emph{how to construct} bearing rigid networks is still unsolved to a large extent.

Construction of bearing rigid networks is practically important for the design of the configurations and interaction topologies of sensor networks or multi-robot formations. The first contribution of this paper is to show that the key to construct bearing rigid networks is to construct appropriate underlying graphs.
In particular, we show that the bearing rigidity of a network is critically determined by its underlying graph rather than the configuration of the nodes. We define a new notion termed \emph{generic bearing rigidity} for graphs. When a graph is generically bearing rigid, then the network is bearing rigid for almost all configurations; otherwise, the network is not bearing rigid for any configuration.
As a result, construction of generically bearing rigid graphs is the key to construct bearing rigid networks.

One of the most well-known graph construction methods is the Henneberg construction, which can be used to construct Laman graphs \cite{Laman1970,WhiteleyHenneberg1985,Whiteley2005Pseudotriangulation,Anderson2008CSM,Jackson2007NotesRigidity,ConnellyBook}.
Laman graphs have played important roles in the distance rigidity theory.\footnote{In order to distinguish from the bearing rigidity theory, we refer the conventional rigidity theory defined based on inter-node distances as \emph{distance rigidity theory}.} In particular, by merely considering generic configurations, a network is minimally distance rigid if and only if the underlying graph is Laman \cite{Laman1970,WhiteleyHenneberg1985,Whiteley2005Pseudotriangulation,Anderson2008CSM,Jackson2007NotesRigidity,ConnellyBook}. This result is known as Laman's Theorem. It is notable that Laman's Theorem is valid merely in two dimensions and a similar result does not exist in higher dimensions.

In this paper, we show that Laman graphs also play important roles in the bearing rigidity theory. The main contribution of this paper is to prove that Laman graphs are generically bearing rigid in arbitrary dimensions. As a result, if the underlying graph of a network is Laman, the network is bearing rigid for almost all configurations in arbitrary dimensions. Since a Laman graph has $2n-3$ edges ($n$ denotes the number of nodes), it is implied that $2n-3$ edges are sufficient to ensure the bearing rigidity of a network in an arbitrary dimension. Furthermore, we show that being Laman is merely sufficient but not necessary for a graph to be generically bearing rigid. A counterexample shows that graphs with less than $2n-3$ edges (hence not Laman) may still be generically bearing rigid. However, if restricting to the two-dimensional plane, we can prove that Laman graphs are both necessary and sufficient for generic bearing rigidity.
Finally, in our previous work \cite{TrinhTAC2016}, the Henneberg construction method has been utilized to construct a special type of directed networks where each node has at most two outgoing edges. In the present paper we consider general Laman graphs without any restrictions.

\section{Preliminaries}

\subsection{Notations of Networks}

Let $\G=(\V,\E)$ be a graph consisting of the vertex set $\V=\{1,\dots,n\}$ and the edge set $\E\subseteq\V\times\V$. In this paper, we only consider undirected graphs where $(i,j)\in\E\Leftrightarrow(j,i)\in\E$. If $(i,j)\in\E$, then $i$ and $j$ are adjacent and vertex $j$ is the neighbor of vertex $i$. The set of the neighbors of vertex $i$ is denoted as $\N_i=\{j\in\V: (i,j)\in\E\}$. A graph $\G_s=(\V_s,\E_s)$ is called a subgraph of $\G$ if $\V_s\subseteq\V$ and $\E_s\subseteq(\V_s\times\V_s)\cap\E$. A subgraph is called a spanning subgraph if it is connected and $\V_s=\V$.

Consider $n$ nodes in $\R^d$ where $n\ge2$, $d\ge2$. Let $p_i\in\R^d$ be the position of node $i$ and $p=[p_1^T,\dots,p_n^T]^T\in\mathbb{R}^{dn}$ be the configuration of these nodes. Assume $p_i\ne p_j$ for all $i$ and $j$ throughout this paper.
A network, denoted as $(\G,p)$, is the graph $\G$ with its vertices mapped to the points $\{p_i\}_{i=1}^n$.
A network may also be called framework or formation under different circumstances.
For edge $(i,j)\in\E$, let
\begin{align*}
g_{ij}=\frac{p_j-p_i}{\|p_j-p_i\|}
\end{align*}
be the unit vector pointing from $p_i$ to $p_j$. The unit vector $g_{ij}$ represents the relative bearing of node $j$ with respect to node $i$.
For $g_{ij}$, define $P: \R^d\rightarrow\R^{d\times d}$ as
\begin{align*}
    P(g_{ij}) =I_d - g_{ij}g_{ij}^T ,
\end{align*}
where $I_d\in\R^{d\times d}$ is the identity matrix. For notational simplicity, we denote $P_{g_{ij}}=P(g_{ij})$.
The matrix $P_{g_{ij}}$ is an orthogonal projection matrix that geometrically projects any vector onto the orthogonal compliment of $g_{ij}$.
It can be verified that $P_{g_{ij}}^T =P_{g_{ij}}$, $P_{g_{ij}}^2=P_{g_{ij}}$, and $P_{g_{ij}}$ is positive semi-definite.
Since $\Null(P_{g_{ij}})=\myspan\{g_{ij}\}$,  for any vector $x\in\R^d$, $P_{g_{ij}}x=0$ if and only if $x$ is parallel to $g_{ij}$.
In this paper $\Null(\cdot)$ and $\Range(\cdot)$ denote the null and range space of a matrix, respectively.
Let $\one_n\triangleq[1,\dots,1]^T\in\R^n$ and $\|\cdot\|$ be the Euclidian norm of a vector or the spectral norm of a matrix, and $\otimes$ be the Kronecker product.
\subsection{Bearing Laplacian and Bearing Rigidity}
We next introduce an important matrix that will be used throughout the paper.
For $(\G,p)$, let $\L\in\R^{dn\times dn}$ be the bearing Laplacian with its $ij$th subblock matrix as \cite{zhao2015NetLocalization}
\begin{align*}
[\L]_{ij}=\left\{
  \begin{array}{ll}
      \mathbf{0}_{d\times d}, & i\ne j, (i,j)\notin\E, \\
      -P_{g_{ij}}, & i\ne j, (i,j)\in\E, \\ 
      \sum_{k\in\N_i}P_{g_{ik}}, & i=j, i\in\V. \\
  \end{array}
\right.
\end{align*}
The bearing Laplacian $\L$ is a matrix-weighted graph Laplacian. It is jointly determined by the underlying graph and the inter-neighbor bearings of the network. It is a symmetric matrix since the graph is assumed to be undirected. For any network, the bearing Laplacian is positive semi-definite because for any $x=[x_1^T ,\dots,x_n^T ]^T \in\R^{dn}$
\begin{align*}
x^T \L x=\frac{1}{2}\sum_{(i,j)\in\E} (x_i-x_j)^TP_{g_{ij}}(x_i-x_j)\ge0.
\end{align*}
For any network, we always have $\rank(\L)\le dn-d-1$ and $\myspan\{\one\otimes I_d,p\}\subseteq\Null(\L)$ \cite[Lem~2]{zhao2015NetLocalization}.

A network is \emph{infinitesimally bearing rigid} if and only if the positions of the nodes in the network can be uniquely determined up to a translational and scaling factor (in other words, the shape of the network can be uniquely determined). The formal definition of infinitesimal bearing rigidity can be found in \cite[Def~5]{zhao2014TACBearing}. Although there exist other types of bearing rigidity such as bearing rigidity and global bearing rigidity, they are not of interest for this paper. For the sake of simplicity, infinitesimal bearing rigidity will be simply referred to as \emph{bearing rigidity} in this paper. A necessary and sufficient condition of bearing rigidity is given below.

\begin{lemma}[\textbf{Condition of Bearing Rigidity \cite[Lem~2]{zhao2015NetLocalization}}]\label{lemma_NSConditionForBearingRigidity}
A network $(\G,p)$ in $\R^d$ is bearing rigid if and only if $\rank(\L)=dn-d-1$ or equivalently $\Null(\L)=\myspan\{\one_n\otimes I_d,p\}$.
\end{lemma}

The condition in Lemma~\ref{lemma_NSConditionForBearingRigidity} provides a convenient way to examine the bearing rigidity of a given network. It will be used later to analyze the construction of bearing rigid networks.

\section{Generic Bearing Rigidity of Graphs}

In this section, we show that the bearing rigidity of a network is a generic property that is critically determined by the underlying graph rather than the configuration. We first define the following notion that will be used throughout the paper.

\begin{definition}[\textbf{Generically Bearing Rigid Graphs}]
A graph $\G$ is generically bearing rigid in $\R^d$ if there exists at least one configuration $p$ in $\R^d$ such that $(\G,p)$ is bearing rigid.
\end{definition}

Generically bearing rigid graphs have the following properties.

\begin{lemma}[\textbf{Density of Generically Bearing Rigid Graphs}]\label{lemma_propertyGenericRigidGraphs}
If $\G$ is generically bearing rigid in $\R^d$, then $(\G,p)$ is bearing rigid for almost all $p$ in $\R^d$ in the sense that the set of $p$ where $(\G,p)$ is not bearing rigid is of measure zero. Moreover, for any configuration $p_0$ and any small constant $\epsilon>0$, there always exists a configuration $p$ such that $(\G,p)$ is bearing rigid and $\|p-p_0\|<\epsilon$.
\end{lemma}
\begin{proof}
Let $\Omega$ be the set of $p$ where $\rank(\L)<dn-d-1$. Suppose $f(p)$ is the vector consisting of all the $(dn-d-2)\times (dn-d-2)$ minors of $\L$. Then, $\Omega$ is the set of solutions to $f(p)=0$. Although the elements of $p$ appear on the denominators in the projection matrices in $\L$, the equation $f(p)=0$ can be converted to a set of polynomial equations of $p$ by multiplying the denominators on both sides of $f(p)=0$.
As a result, $\Omega$ is an algebraic set and hence it is either the entire space or of measure zero \cite{zeroSetPolynomial}. Since there exists $p$ such that $(\G,p)$ is bearing rigid, $\Omega$ is not the entire space, then it is of measure zero and consequently $(\G,p)$ is bearing rigid for almost all $p$.

For the sake of completeness, we next present an elementary proof of the density of bearing rigid networks. Since $\G$ is generically bearing rigid, there exists $p_1$ such that $(\G,p_1)$ is bearing rigid. For the given configuration $p_0$, define
\begin{align*}
p_\alpha=(1-\alpha)p_0+\alpha p_1.
\end{align*}
When $\alpha=0$, $p_\alpha=p_0$; when $\alpha=1$, $p_\alpha=p_1$. For any $\epsilon>0$, there always exists a sufficiently small $\alpha_\epsilon$ such that $\|p_{\alpha}-p_0\|<\epsilon$ for all $\alpha\in(0,\alpha_\epsilon)$. Let $f(\alpha)$ be the vector consisting of all the $(dn-d-2)\times (dn-d-2)$ minors of $\L$ of the network $(\G,p_\alpha)$. Then $f(\alpha)\ne0$ if and only if $(\G,p_\alpha)$ is bearing rigid. Since $f(1)\ne0$, $f(\alpha)$ is not identically zero. Since $f(\alpha)=0$ can be converted to a set of polynomial equations of $\alpha$, $f(\alpha)=0$ has finite zero roots. As a result, there always exists $\alpha_1\in(0,\alpha_\epsilon)$ such that $f(\alpha_1)\ne0$. Then, the network $(\G,p_{\alpha_1})$ is bearing rigid and satisfies $\|p_{\alpha_1}-p_0\|<\epsilon$.
\end{proof}

If a graph is not generically bearing rigid, there does not exist any configuration such that the network is bearing rigid. This is implied by the definition of generic bearing rigidity. See Fig.~\ref{fig_demoGenericGraph}(a) for an illustration.
If a graph is generically bearing rigid, the corresponding networks are bearing rigid for all configurations except some special ones that form a set of measure zero. This is implied by Lemma~\ref{lemma_propertyGenericRigidGraphs}. See Fig.~\ref{fig_demoGenericGraph}(b) for an illustration. If a network is not bearing rigid but its graph is generically bearing rigid, then there always exists a sufficiently small perturbation of the configuration that can make the network bearing rigid.

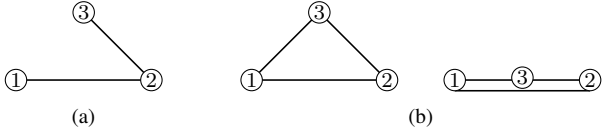
\begin{figure}[t]
  \centering
  \def\myscale{0.45}
  \def\length{2}
\def\radius{9pt}

\subfloat[]{
\begin{tikzpicture}[scale=\myscale]
\coordinate (1) at (-\length,0);
\coordinate (2) at (\length,0);
\coordinate (3) at (0,\length);
\draw [semithick,] (1)--(2)--(3);
\draw [fill=white](1) circle [radius=\radius];
\draw (1) node[] {\footnotesize$1$};
\draw [fill=white](2) circle [radius=\radius];
\draw (2) node[] {\footnotesize$2$};
\draw [fill=white](3) circle [radius=\radius];
\draw (3) node[] {\footnotesize$3$};
\end{tikzpicture}
}\qquad
\subfloat[]{
\begin{tikzpicture}[scale=\myscale]
\coordinate (1) at (-\length,0);
\coordinate (2) at (\length,0);
\coordinate (3) at (0,\length);
\draw [semithick,] (1)--(2)--(3)--cycle;
\draw [fill=white](1) circle [radius=\radius];
\draw (1) node[] {\footnotesize$1$};
\draw [fill=white](2) circle [radius=\radius];
\draw (2) node[] {\footnotesize$2$};
\draw [fill=white](3) circle [radius=\radius];
\draw (3) node[] {\footnotesize$3$};
\begin{scope}[shift={(3*\length,0)},rotate=0]
\coordinate (1) at (-\length,0);
\coordinate (2) at (\length,0);
\coordinate (3) at (0,0);
\draw [semithick,] (1)--(3);
\draw [semithick,] (2)--(3);
\draw [semithick,] ($(1)+(0,-\radius)$)--($(2)+(0,-\radius)$);
\draw [fill=white](1) circle [radius=\radius];
\draw (1) node[] {\footnotesize$1$};
\draw [fill=white](2) circle [radius=\radius];
\draw (2) node[] {\footnotesize$2$};
\draw [fill=white]($(3)+(0,2pt)$) circle [radius=\radius];
\draw ($(3)+(0,2pt)$) node[] {\footnotesize$3$};
\end{scope}
\end{tikzpicture}
}
  \caption{The graph of the network in (a) is not generically bearing rigid. As a result, the network is not bearing rigid for any configuration. The graph of the networks in (b) is generically bearing rigid. The network is bearing rigid for almost all configurations except those where the three nodes are collinear.}
  \vskip-10pt
  \label{fig_demoGenericGraph}
\end{figure}

\section{Construction of Generically Bearing Rigid Graphs}

In the preceding section, we have shown that the key to construct bearing rigid networks is to construct generically bearing rigid graphs.
In this section, we address how to construct generically bearing rigid graphs. We start from the definition of an important type of graphs.

\begin{definition}[\textbf{Laman Graphs \cite{WhiteleyHenneberg1985,Whiteley2005Pseudotriangulation,Jackson2007NotesRigidity,Anderson2008CSM,ConnellyBook}}]
A graph $\G=(\V,\E)$ is Laman if $|\E|=2|\V|-3$ and every subset of $k\ge2$ vertices spans at most $2k-3$ edges.
\end{definition}

The above is a combinatorial definition of Laman graphs. Its intuition is that the edges should be distributed evenly in a Laman graph.
Laman graphs may also be characterized by the Henneberg construction as described below.

\begin{definition}[\textbf{Henneberg Construction \cite{WhiteleyHenneberg1985,Whiteley2005Pseudotriangulation,Jackson2007NotesRigidity,Anderson2008CSM,ConnellyBook}}]\label{def_constructionmethod}
Given a graph $\G=(\V,\E)$, a new graph $\G'=(\V',\E')$ is formed by adding a new vertex $v$ to $\G$ and performing one of the following two operations:
\begin{enumerate}[(a)]
\item \emph{Vertex addition}: connect vertex $v$ to any two existing vertices $i,j\in\V$. In this case, $\V'=\V\cup \{v\}$ and $\E'=\E\cup\{(v,i),(v,j)\}$. See Fig.~\ref{fig_demoHennebergConstruction}(a) for an illustration.
\item \emph{Edge splitting}: consider three vertices $i,j,k\in\V$ with $(i,j)\in\E$ and connect vertex $v$ to $i,j,k$ and delete $(i,j)$. In this case, $\V'=\V\cup \{v\}$ and $\E'=\E\cup\{(v,i),(v,j),(v,k)\}\setminus\{(i,j)\}$. See Fig.~\ref{fig_demoHennebergConstruction}(b) for an illustration.
\end{enumerate}
\end{definition}

\begin{figure}[h]
  \centering
  \def\myscale{0.45}
  \def\length{2.5}
\def\radius{10pt}
\def\linkLength{1}
\def\crossLength{8pt}
\subfloat[Vertex addition]{
\begin{tikzpicture}[scale=\myscale]
\coordinate (v) at (0,\length);
\coordinate (i) at (-\length,0);
\coordinate (j) at (\length/3,-\length/2.5);
\draw [blue,thick] (i)--(v)--(j);
\draw ($(i)$)--($(i)+(\linkLength,-\linkLength/3)$);
\draw ($(i)$)--($(i)+(-\linkLength/2,-\linkLength)$);
\draw ($(j)$)--($(j)+(-\linkLength,\linkLength/3)$);
\draw ($(j)$)--($(j)+(-\linkLength/2,-\linkLength)$);
\draw [fill=white,draw=blue](v) circle [radius=\radius];
\draw (v) node[blue] {\footnotesize$v$};
\draw [fill=white](i) circle [radius=\radius];
\draw (i) node[] {\footnotesize$i$};
\draw [fill=white](j) circle [radius=\radius];
\draw (j) node[] {\footnotesize$j$};
\draw (-\length/2,-\length/1.3) node[] {\footnotesize$\G$};
\node [draw,cloud,cloud puffs=20,cloud puff arc=120,aspect=2,minimum height=1.6cm,minimum width=3cm,thin,dotted] (myCloud) at (-0.7cm,-0.7cm) {};
\end{tikzpicture}
}
\qquad
\subfloat[Edge splitting]{
\begin{tikzpicture}[scale=\myscale]
\coordinate (v) at (0,\length);
\coordinate (i) at (-\length,0);
\coordinate (j) at (0,-\length/2);
\coordinate (k) at (\length/2,0);
\draw [blue,thick] (i)--(v)--(j);
\draw [blue,thick] (v)--(k);
\draw [] (i)--(j);
\draw [red] ($(i)!0.5!(j)+(\crossLength,\crossLength)$)--($(i)!0.5!(j)+(-\crossLength,-\crossLength)$);
\draw [red] ($(i)!0.5!(j)+(-\crossLength,\crossLength)$)--($(i)!0.5!(j)+(\crossLength,-\crossLength)$);
\draw ($(i)$)--($(i)+(-\linkLength/2,-\linkLength)$);
\draw ($(j)$)--($(j)+(-\linkLength/2,-\linkLength)$);
\draw ($(k)$)--($(k)+(-\linkLength,-\linkLength/3)$);
\draw ($(k)$)--($(k)+(\linkLength/2,-\linkLength)$);
\draw [fill=white,draw=blue](v) circle [radius=\radius];
\draw (v) node[blue] {\footnotesize$v$};
\draw [fill=white](i) circle [radius=\radius];
\draw (i) node[] {\footnotesize$i$};
\draw [fill=white](j) circle [radius=\radius];
\draw (j) node[] {\footnotesize$j$};
\draw [fill=white](k) circle [radius=\radius];
\draw (k) node[] {\footnotesize$k$};
\draw (-\length/2,-\length/1.3) node[] {\footnotesize$\G$};
\node [draw,cloud,cloud puffs=20,cloud puff arc=120,aspect=2,minimum height=1.6cm,minimum width=3cm,thin,dotted] (myCloud) at (-0.5cm,-0.7cm) {};
\end{tikzpicture}}
  \caption{An illustration of the Henneberg construction.}
  \vskip-10pt
  \label{fig_demoHennebergConstruction}
\end{figure}
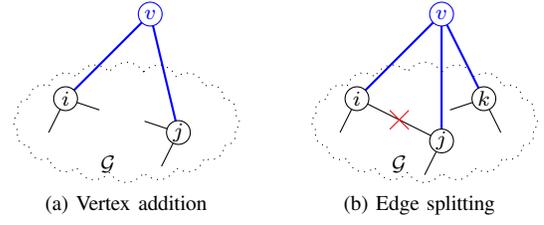

A Henneberg construction starting from an edge connecting two vertices will result in a Laman graph \cite{WhiteleyHenneberg1985,Whiteley2005Pseudotriangulation,Jackson2007NotesRigidity}. The converse is also true. That is if a graph is Laman, then it can be generated by a Henneberg construction \cite[Lem~2]{Whiteley2005Pseudotriangulation}.

The following theorem is the main result in this paper.

\begin{theorem}[\textbf{Generic Bearing Rigidity of Laman Graphs}]\label{thm_HennebergConstruction}
If $\G$ is a Laman graph, then $\G$ is generically bearing rigid in $\R^d$ for any $d\ge2$.
\end{theorem}
\begin{proof}
The proof requires some additional lemmas and is deferred to Section~\ref{section_proofofMainResults}.
\end{proof}

Theorem~\ref{thm_HennebergConstruction} indicates that a network with a Laman graph is bearing rigid for almost all configurations in $\R^d$ for any $d\ge2$.
It also indicates that $2n-3$ edges are sufficient to guarantee the bearing rigidity of a network in an arbitrary dimension since a Laman graph has $2n-3$ edges. For example, every network in Fig.~\ref{fig_bearingRigidNetworkConstructionEachStep} is bearing rigid in $\R^3$ and has merely $2n-3$ edges.
Figure~\ref{fig_bearingRigidNetworkConstructionEachStep} shows all the steps to construct a three-dimensional bearing rigid network. 

\begin{figure}[t]
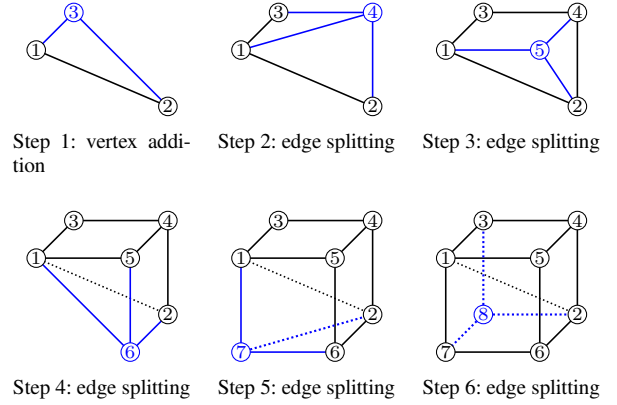

  \centering
  \def\myscale{0.5}
  \include{figure_tikz_constructCube}
  \caption{The procedure to construct a three-dimensional bearing rigid network. The number of edges in this network is equal to $2n-3=13$.}
  \vspace{-10pt}
  \label{fig_bearingRigidNetworkConstructionEachStep}
\end{figure}

The next result shows that adding edges to a Laman graph preserves generic bearing rigidity.

\begin{corollary}\label{corollary_LamanGraphAddEdge}
If $\G$ contains a Laman spanning subgraph, then $\G$ is generically bearing rigid in $\R^d$ for any $d\ge2$.
\end{corollary}
\begin{proof}
The proof requires some additional lemmas and is deferred to Section~\ref{section_proofofMainResults}.
\end{proof}

While Theorem~\ref{thm_HennebergConstruction} indicates that Laman graphs are generically bearing rigid, a natural question that follows is whether generically bearing rigid graphs are also Laman. The answer may be negative. A counterexample is given in Fig.~\ref{fig_specialRigidExample}. The cyclic graph in this example is generically bearing rigid in $\R^3$ because the configuration in Fig.~\ref{fig_specialRigidExample}(b) makes the network bearing rigid. However, this cyclic graph is \emph{not} Laman because the edge number of the graph is $4$, which is less than $2n-3=5$ (a Laman graph must have $2n-3$ edges).
This example also demonstrates that $2n-3$ is not the minimum number of edges required to ensure bearing rigidity. A discussion on this example is given in the conclusion section.

\begin{figure}[t]
  \centering
  \def\myscale{0.45}
  \def\length{2}
\def\radius{8pt}

\subfloat[]{
\begin{tikzpicture}[scale=\myscale]
\coordinate (origin) at (0,0);
\coordinate (x) at (-1.1*\length,-1.1*\length);
\coordinate (y) at (2.1*\length,0);
\coordinate (z) at (0,1.6*\length);
\draw [->,very thin,densely dashed] (origin)--(x) node[left]{\scriptsize$x$};
\draw [->,very thin,densely dashed] (origin)--(y) node[right]{\scriptsize$y$};
\draw [->,very thin,densely dashed] (origin)--(z) node[right]{\scriptsize$z$};
\coordinate (1) at (-0.7*\length,-0.7*\length);
\coordinate (2) at (0.8*\length,-0.7*\length);
\coordinate (3) at (1.5*\length,0);
\coordinate (4) at (0,0);
\draw [semithick] (1)--(2)--(3)--(4)--cycle;
\draw [fill=white](1) circle [radius=\radius];
\draw (1) node[] {\scriptsize$1$};
\draw [fill=white](2) circle [radius=\radius];
\draw (2) node[] {\scriptsize$2$};
\draw [fill=white](3) circle [radius=\radius];
\draw (3) node[] {\scriptsize$3$};
\draw [fill=white](4) circle [radius=\radius];
\draw (4) node[] {\scriptsize$4$};
\end{tikzpicture}
}\quad
\subfloat[]{
\begin{tikzpicture}[scale=\myscale]
\coordinate (origin) at (0,0);
\coordinate (x) at (-1.1*\length,-1.1*\length);
\coordinate (y) at (2.1*\length,0);
\coordinate (z) at (0,1.6*\length);
\draw [->,very thin,densely dashed] (origin)--(x) node[left]{\scriptsize$x$};
\draw [->,very thin,densely dashed] (origin)--(y) node[right]{\scriptsize$y$};
\draw [->,very thin,densely dashed] (origin)--(z) node[right]{\scriptsize$z$};
\coordinate (1) at (-0.7*\length,-0.7*\length);
\coordinate (2) at (0.8*\length,-0.7*\length);
\coordinate (3) at (1.5*\length,0);
\coordinate (4) at (0,\length);
\draw [semithick] (1)--(2)--(3)--(4)--cycle;
\draw [fill=white](1) circle [radius=\radius];
\draw (1) node[] {\scriptsize$1$};
\draw [fill=white](2) circle [radius=\radius];
\draw (2) node[] {\scriptsize$2$};
\draw [fill=white](3) circle [radius=\radius];
\draw (3) node[] {\scriptsize$3$};
\draw [fill=white](4) circle [radius=\radius];
\draw (4) node[] {\scriptsize$4$};
\end{tikzpicture}
}
  \caption{The configuration (a) is in the $x$--$y$ plane and the network is not bearing rigid. The configuration (b) is three-dimensional and the network is bearing rigid. It can be verified that $\rank(\L)=dn-d-1$ for the configuration in (b).}
  \vskip-10pt
  \label{fig_specialRigidExample}
\end{figure}
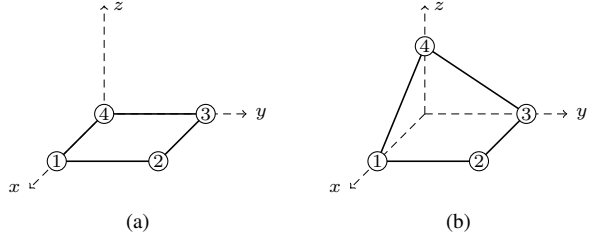

As indicated by the example in Fig.~\ref{fig_specialRigidExample}, not all generically bearing rigid graphs are Laman or contain Laman spanning subgraphs. However, if we restrict to $\R^2$, then Laman is both necessary and sufficient for generic bearing rigidity.

\begin{theorem}
A graph $\G$ is generically bearing rigid in $\R^2$ if and only if the graph contains a Laman spanning subgraph.
\end{theorem}
\begin{proof}
The sufficiency follows from Theorem~\ref{thm_HennebergConstruction} and Corollary~\ref{corollary_LamanGraphAddEdge}. To prove necessity, we need some notions in the distance rigidity theory \cite{Jackson2007NotesRigidity,ConnellyBook} which are omitted here due to space limitations. Since $\G$ is generically bearing rigid in $\R^2$, there exists $p$ such that $(\G,p)$ is infinitesimally bearing rigid in $\R^2$. Since infinitesimal bearing rigidity and infinitesimal distance rigidity imply each other in $\R^2$ \cite[Thm~8]{zhao2014TACBearing}, we know $(\G,p)$ is infinitesimally distance rigid in $\R^2$. Now we consider two cases. In the first case where $\G$ has exactly $2n-3$ edges, the distance rigidity matrix of $(\G,p)$ has full row rank and consequently the graph is Laman by \cite[Lem~2.3]{Jackson2007NotesRigidity} (note that \cite[Lem~2.3]{Jackson2007NotesRigidity} is the necessary part of Laman's Theorem). In the second case where $\G$ has more than $2n-3$ edges, the distance rigidity matrix has its rank equal to $2n-3$ though it is not of full row rank any more. There must exist $2n-3$ linearly independent rows in the distance rigidity matrix. These $2n-3$ rows correspond to a spanning subgraph with $2n-3$ edges. Since the distance rigidity matrix of this spanning subgraph is of full row rank, the subgraph is Laman by \cite[Lem~2.3]{Jackson2007NotesRigidity}.
\end{proof}

\section{Proof of the Main Results}\label{section_proofofMainResults}

In this section, we prove Theorem~\ref{thm_HennebergConstruction} and Corollary~\ref{corollary_LamanGraphAddEdge}. To do that, we need first prove some lemmas.

\begin{lemma}\label{lemma_singularityOfProjSum}
For any nonzero vectors $x_1,\dots,x_m\in\R^d$ where $m\ge2,d\ge2$, the matrix $E=\sum_{i=1}^m P_{x_i}\in\R^{d\times d}$ is nonsingular if and only if at least two of $x_1,\dots,x_m$ are not collinear.
\end{lemma}
\begin{proof}
Let $F=[P_{x_1},\dots,P_{x_m}]^T\in\R^{dm\times d}$. Since $P_{x_i}^2=P_{x_i}$, we know $E=F^TF $. Then $Ez=0\Leftrightarrow z^TEz=0\Leftrightarrow Fz=0\Leftrightarrow P_{x_1}z=\dots=P_{x_m}z=0$. Since $P_{x_i}z=0$ if and only if $z$ is collinear with $x_i$, $E$ is singular if and only if $x_1,\dots,x_m$ are all collinear.
\end{proof}

\begin{lemma}\label{lemma_eigNullOfH}
For any nonzero vectors $x_1,\dots,x_m\in\R^d$, at least two of which are not collinear, the $dm\times dm$ matrix
$$H=D-FE^{-1}F^T ,$$
where
\begin{align*}
\small{
D=\left[
  \begin{array}{ccc}
    P_{x_1} & \dots & 0\\
    \vdots &  \ddots & \vdots\\
    0 & \dots & P_{x_m}
  \end{array}
\right],
F=\left[
  \begin{array}{c}
    P_{x_1}\\
    \vdots \\
    P_{x_m}\\
  \end{array}
\right],
E=\sum_{i=1}^m P_{x_i},}
\end{align*}
is a symmetric positive semi-definite matrix with $m+d$ eigenvalues equal to zero and $dm-m-d$ eigenvalues equal to one.
\end{lemma}
\begin{proof}
Let
\begin{align}\label{eq_orthogonalBasisNullH}
N=
\left[
  \begin{array}{ccccc}
    x_1 & 0 & \dots & 0           & P_{x_1}\\
    0 & x_2 & \dots & 0           & P_{x_2}\\
    \vdots & \vdots & \ddots & 0  & \vdots\\
    0 & 0 & \dots & x_m              & P_{x_m}\\
  \end{array}
\right]\in\R^{dm\times(m+d)}.
\end{align}
Since $P_{x_i}x_i=0$ and $P_{x_i}P_{x_i}=P_{x_i}$, it can be verified that $HN=0$ and hence $\Range(N)\subseteq\Null(H)$.
Since $x_i$'s are not collinear, the last $d$ columns of $N$ are of full column rank and hence $N$ is of full column rank.

Let $z=[z_1^T ,\dots,z_2^T ]^T \in\R^{dm}$ be a vector orthogonal to $\Range(N)$, which means $z^TN=0$ and equivalently
\begin{align}\label{eq_perpRangeNCondition}
x_1^Tz_1=0,\quad\dots,\quad x_m^Tz_m=0,\quad \sum_{i=1}^m P_{x_i}z_i=0.
\end{align}
It follows from $x_i^Tz_i=0$ that $P_{x_i}z_i=z_i$. Then,
\begin{align*}
Hz
=\left[
   \begin{array}{c}
     P_{x_1}z_1 \\
     \vdots \\
     P_{x_2}z_2 \\
   \end{array}
 \right]-FE^{-1}\sum_{i=1}^m P_{x_i}z_i
=\left[
   \begin{array}{c}
     z_1 \\
     \vdots \\
     z_2 \\
   \end{array}
 \right]-0=z.
\end{align*}
The above equation implies two conclusions. First, every $z\perp\Range(N)$ is not in the null space of $H$ and hence $\Range(N)=\Null(H)$.
Second, every $z\perp\Range(N)$ is an eigenvector of $H$ and the corresponding eigenvalue is 1.
\end{proof}

\begin{lemma}\label{lemma_rankOfPSDMatSum}
If $A,B\in\R^{m\times m}$ are positive semi-definite, $\rank(A+B)\ge\max\{\rank(A),\rank(B)\}$.
\end{lemma}
\begin{proof}
The proof is straightforward and omitted due to space limitation.
\end{proof}

Now we prove the main result of this paper.

\begin{proof}[\textbf{Proof of Theorem~\ref{thm_HennebergConstruction}}]
Every Laman graph can be constructed iteratively by the Henneberg construction starting from a simple graph of two vertices and one edge. This simple graph is generically bearing rigid in $\R^d$ for any $d\ge2$. The rest of the proof is to show that if $\G$ is generically bearing rigid in $\R^d$ then $\G'$ remains generically bearing rigid in $\R^d$ after either operation in the Henneberg construction. Then the theorem can be proved by induction.

Since $\G$ is generically bearing rigid, there exists a configuration $p$ such that $(\G,p)$ is bearing rigid. Let $p_v\in\R^d$ be the position of the new node and $p'=[p^T ,p_v^T ]^T \in\R^{d(n+1)}$.  Let $\L$ and $\L'$ be the bearing Laplacian matrices of $(\G,p)$ and $(\G',p')$, respectively. Then, $\rank(\L)=dn-d-1$ by Lemma~\ref{lemma_NSConditionForBearingRigidity}. Our objective is to prove there exists $p_v$ such that $\rank(\L')=d(n+1)-d-1$.

\emph{Case 1: Vertex Addition.}  We first consider the case where $\G'$ is obtained from $\G$ by the vertex-addition operation in Definition~\ref{def_constructionmethod}. Partition $\L$ into
\begin{align*}
\L=\left[
     \begin{array}{cc}
       \L_{11} & \L_{12} \\
       \L_{21} & \L_{22} \\
     \end{array}
   \right],
\end{align*}
where $\L_{22}\in\R^{2d\times 2d}$ corresponds to nodes $i,j$. Then $\L'$ can be expressed as
\begin{align*}
\L'=\left[
     \begin{array}{cc:c}
       \L_{11} & \L_{12} & 0\\
       \L_{21} & \L_{22}+D & F\\
       \hdashline
       0       & F^T        & E \\
     \end{array}
   \right],
\end{align*}
where
\begin{align*}
D&=\left[
     \begin{array}{cc}
       P_{g_{iv}} & 0\\
       0 & P_{g_{jv}}\\
     \end{array}
   \right]\in\R^{2d\times 2d},\\
F&=\left[
     \begin{array}{cc}
       -P_{g_{iv}}\\
       -P_{g_{jv}}\\
     \end{array}
   \right]\in\R^{2d\times d},\\
E&=P_{g_{iv}}+P_{g_{jv}}\in\R^{d\times d},
\end{align*}
and $0$ denotes zero matrices with appropriate dimensions.

In order to show that $\G'$ is generically bearing rigid, we only need to find at least one configuration $p'$ such that $(\G',p')$ is bearing rigid. In this direction, consider the case where $p_i,p_j,p_v$ are not collinear. Then $g_{iv}$ and $g_{jv}$ are not collinear and hence $E$ is nonsingular by Lemma~\ref{lemma_singularityOfProjSum}.
By the properties of $2\times2$ block matrices, we have
\begin{align}
&\rank(\L')\nonumber\\
&=\rank(E)\nonumber\\
&\quad+\rank
\left(
\L+\left[
     \begin{array}{ccc}
       0 & 0\\
       0 & D\\
     \end{array}
   \right]
-\left[
     \begin{array}{cc}
       0\\
       F\\
     \end{array}
   \right]E^{-1}\left[
     \begin{array}{cc}
       0 & F^T
     \end{array}
   \right]
\right)\nonumber\\
&=\rank(E)+\rank
\left(
\L+\left[
     \begin{array}{ccc}
       0 & 0\\
       0 & D-FE^{-1}F^T \\
     \end{array}
   \right]
\right)\nonumber\\
&:=\rank(E)+\rank(\L+\bar{H}).
\end{align}
According to Lemma~\ref{lemma_eigNullOfH}, we know $D-FE^{-1}F^T $ is positive semi-definite and so is $\bar{H}$. Then, $\rank(\L+\bar{H})\ge\rank(\L)=dn-d-1$ by Lemma~\ref{lemma_rankOfPSDMatSum}. Since $\rank(E)=d$, we know $\rank(\L')=\rank(E)+\rank(\L+\bar{H})\ge d+dn-d-1=d(n+1)-d-1$. Since $\rank(\L')\le d(n+1)-d-1$, we obtain $\rank(\L')= d(n+1)-d-1$. Therefore, $(\G',p')$ is bearing rigid and hence $\G'$ is generically bearing rigid.

\emph{Case 2: Edge Splitting.} We now consider the case where $\G'$ is obtained from $\G$ by the edge-splitting operation in Definition~\ref{def_constructionmethod}.
Since $\G$ is generically bearing rigid, there always exists a configuration where $p_i$, $p_j$, and $p_k$ are not collinear such that $(\G,p)$ is bearing rigid.\footnote{If $p_i$, $p_j$, and $p_k$ are non-collinear but $(\G,p)$ is not bearing rigid, we can always apply a sufficiently small perturbation on the configuration to make the network bearing rigid according to Lemma~\ref{lemma_propertyGenericRigidGraphs}. Since the perturbation may be sufficiently small, $p_i$, $p_j$, and $p_k$ can be preserved to be non-collinear.} By placing the new node $p_v$ in the middle of $p_i$ and $p_j$, we obtain $(\G',p')$ (see Fig.~\ref{fig_edgeSplittingCollinear}(a)). By deleting $(v,j)$ from $\G'$ and adding $(i,j)$ to $\G'$, we obtain another network $(\G^*,p')$ where $\G^*=(\V,\E^*)$ and $\E^*=\E'\cup\{(i,j)\}\setminus\{(v,j)\}$ (see Fig.~\ref{fig_edgeSplittingCollinear}(b)).

\begin{figure}[t]
  \centering
  \def\myscale{0.45}
  \def\length{2}
\def\radius{10pt}

\subfloat[Network $(\G',p')$]{
\begin{tikzpicture}[scale=\myscale]
\coordinate (v) at (0,0);
\coordinate (i) at (-\length,\length);
\coordinate (j) at (\length,-\length);
\coordinate (k) at (-2*\length,-\length);
\draw [] (v)--(i);
\draw [] (v)--(j);
\draw [] (v)--(k);
\draw [fill=white](v) circle [radius=\radius];
\draw (v) node[] {\footnotesize$v$};
\draw [fill=white](i) circle [radius=\radius];
\draw (i) node[] {\footnotesize$i$};
\draw [fill=white](j) circle [radius=\radius];
\draw (j) node[] {\footnotesize$j$};
\draw [fill=white](k) circle [radius=\radius];
\draw (k) node[] {\footnotesize$k$};
\end{tikzpicture}
}
\qquad
\subfloat[Network $(\G^*,p')$]{
\begin{tikzpicture}[scale=\myscale]
\coordinate (v) at (0,0);
\coordinate (i) at (-\length,\length);
\coordinate (j) at (\length,-\length);
\coordinate (k) at (-2*\length,-\length);
\draw [] (v)--(i);
\draw [] (v)--(k);
\draw ($(i)+(\radius,0)$)--($(j)+(\radius,0)$);
\draw [fill=white](v) circle [radius=\radius];
\draw (v) node[] {\footnotesize$v$};
\draw [fill=white](i) circle [radius=\radius];
\draw (i) node[] {\footnotesize$i$};
\draw [fill=white](j) circle [radius=\radius];
\draw (j) node[] {\footnotesize$j$};
\draw [fill=white](k) circle [radius=\radius];
\draw (k) node[] {\footnotesize$k$};
\end{tikzpicture}
}
  \caption{An illustration of the networks $(\G',p')$ and $(\G^*,p')$ in case 2. Network (a) has $(v,i)$, $(v,j)$, and $(v,k)$. Network (b) has $(v,i)$, $(i,j)$, and $(v,k)$.}
  \vskip-10pt
  \label{fig_edgeSplittingCollinear}
\end{figure}
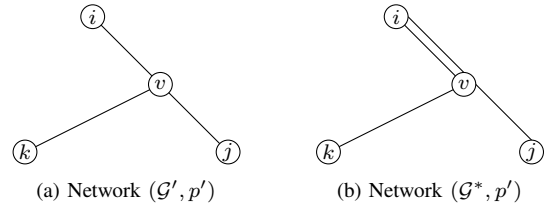
We next show that $(\G',p')$ is bearing rigid if and only if $(\G^*,p')$ is bearing rigid. To do that, let $\L'$ and $\L^*$ be the bearing Laplacian matrices of $(\G',p')$ and $(\G^*,p')$, respectively. For any $x=[x_1^T ,\dots,x_n^T ]^T \in\R^{d(n+1)}$, we have
\begin{align}\label{eq_xL*x}
x^T \L^*x=x^T \L'x&-(x_v-x_j)^TP_{g_{vj}}(x_v-x_j) \nonumber\\
&+(x_i-x_j)^TP_{g_{ij}}(x_i-x_j).
\end{align}
Note that $x^T \L'x=0$ if and only if
\begin{align*}
\sum_{(a,b)\in\E'} (x_a-x_b)^TP_{g_{ab}}(x_a-x_b)=0,
\end{align*}
which implies $(x_a-x_b)^TP_{g_{ab}}(x_a-x_b)=0$ for all $(a,b)\in\E$.
As a result,
\begin{align}\label{eq_xvxjcollinear}
    (x_v-x_i)^TP_{g_{vi}}(x_v-x_i)&=0, \nonumber\\
    (x_v-x_j)^TP_{g_{vj}}(x_v-x_j)&=0.
\end{align}
Since $p_i,p_v,p_j$ are collinear, we know that $g_{vi},g_{vj},g_{ij}$ are parallel to each other. It then follows from the above two equations that $x_v,x_i,x_j$ are collinear. As a result, we have
\begin{align}\label{eq_xixjcollinear}
    (x_i-x_j)^TP_{g_{ij}}(x_i-x_j)=0.
\end{align}
Substituting \eqref{eq_xvxjcollinear} and \eqref{eq_xixjcollinear} into \eqref{eq_xL*x} gives $x^T \L^*x=0$. To summarize, $x^T \L'x=0\Rightarrow x^T \L^*x=0$. Similarly, it can be proved that $x^T \L^*x=0\Rightarrow x^T \L'x=0$ and consequently $x^T \L^*x=0\Leftrightarrow x^T \L'x=0$. As a result, $\L^*$ and $\L'$ have the same rank and consequently $(\G',p')$ is bearing rigid if and only if $(\G^*,p')$ is bearing rigid.

Note that $(\G^*,p')$ is obtained by the vertex addition operation from $(\G,p)$. As we have proved in Case 1, $(\G^*,p')$ is bearing rigid. Therefore, $(\G',p')$ is bearing rigid and consequently $\G'$ is generically bearing rigid.
\end{proof}
\begin{remark}
We would like to point out that there is another simple proof for Theorem~\ref{thm_HennebergConstruction} as outlined below. If a graph $\G$ is Laman, it follows from Laman's Theorem that there exists a configuration $p$ in $\R^2$ such that $(\G,p)$ is infinitesimally distance rigid \cite{Laman1970,WhiteleyHenneberg1985,Whiteley2005Pseudotriangulation,Anderson2008CSM,Jackson2007NotesRigidity,ConnellyBook}. Since a network in $\R^2$ is infinitesimally distance rigid if and only if it is infinitesimally bearing rigid \cite[Thm~8]{zhao2014TACBearing}, $(\G,p)$ is also infinitesimally bearing rigid in $\R^2$. Now lift the network from $\R^2$ to $\R^d$ by changing $p_i\in\R^2$ to $p_i'=[p_i^T ,0]^T \in\R^d$ and suppose $(\G,p')$ is the obtained network in $\R^d$. Since infinitesimal bearing rigidity is invariant to space dimensions \cite[Thm~7]{zhao2014TACBearing}, we know $(\G,p')$ is infinitesimally bearing rigid in $\R^d$. Therefore, $\G$ is generically bearing rigid in $\R^d$. Though simple, the above proof is not self-contained and requires prerequisites in distance rigidity theory.
\end{remark}

We next prove Corollary~\ref{corollary_LamanGraphAddEdge}. To do that, we first show that adding edges to a graph preserves its generic bearing rigidity.

\begin{lemma}\label{lemma_addEdgePreserveRigidity}
Given a graph $\G=(\V,\E)$, suppose $\G'=(\V,\E')$ where $\E'=\E\cup\{(i,j)\}$. If $\G$ is generically bearing rigid, then $\G'$ is also generically bearing rigid.
\end{lemma}
\begin{proof}
Since $\G$ is generically bearing rigid, there exists $p$ such that $(\G,p)$ is bearing rigid. Let $\L$ and $\L'$ be the bearing Laplacian matrices of $(\G,p)$ and $(\G',p)$, respectively. Then, $\rank(\L)=dn-d-1$. Since $\G'$ is obtained by adding one edge to $\G$, we have $\L'=\L+\L_0$ where $\L_0$ is the bearing Laplacian of the network $(\G_0,p)$ where $\G_0=(\V,\E_0)$ and $\E_0=\{(i,j)\}$. Since $\L_0$ is a bearing Laplacian, it is positive semi-definite. It then follows from Lemma~\ref{lemma_rankOfPSDMatSum} that $\rank(\L')=\rank(\L+\L_0)\ge\rank(\L)=dn-d-1$. Since $\rank(\L')\le dn-d-1$, we know $\rank(\L')=dn-d-1$ and consequently $(\G',p)$ is bearing rigid. Therefore, $\G'$ is generically bearing rigid by definition.
\end{proof}

\begin{proof}[\textbf{Proof of Corollary~\ref{corollary_LamanGraphAddEdge}}]
If $\G$ has a Laman spanning subgraph, $\G$ can be obtained by adding edges into the Laman graph. Since a Laman graph is generically bearing rigid, $\G$ is also generically bearing rigid by Lemma~\ref{lemma_addEdgePreserveRigidity}.
\end{proof}

\section{Discussion}
In this section, we give some comments on the results obtained in this paper.

As mentioned in the introduction section, Laman's Theorem in the distance rigidity theory is valid merely in the plane. In three or higher dimensional spaces, extra conditions and more edges are required to guarantee distance rigidity. The fundamental reason is that when lifted to a higher dimension, the degrees of freedom of each node in a network increase whereas the number of constraints posed by inter-neighbor distances does not.
As a comparison, in this paper we have showed that a Laman graph is generically bearing rigid in arbitrary dimensions and $2n-3$ edges would be sufficient to guarantee bearing rigidity in arbitrary dimensions. The fundamental reason is that when lifted to higher dimensions, the number of independent constraints posed by each inter-neighbor bearing also increases.

We now revisit the example in Fig.~\ref{fig_specialRigidExample}. This example may be counterintuitive because it shows that a network is not bearing rigid in a lower dimension yet another network with the same underlying graph is bearing rigid in a higher dimension. This phenomenon may also be explained by the number of independent constraints posed by a bearing. In particular, in order to ensure the cyclic network to be bearing rigid in $\R^2$, the bearings must provide $2n-3=5$ independent constraints. Each bearing in $\R^2$ is equivalent to a bearing angle and hence four bearings can provide at most four independent constraints. Since four is less than $2n-3=5$, the network in $\R^2$ is not bearing rigid. In order to ensure the cyclic network to be bearing rigid in $\R^3$, the bearings must provide $3n-4=8$ independent constraints. Each bearing in $\R^3$ is equivalent to two bearing angles and hence four bearings can provide at most eight independent constraints which is equal to $3n-4$. This is an intuitive way to explain why four bearings ensures the bearing rigidity of the cyclic network in Fig.~\ref{fig_specialRigidExample}(b). The bearing rigidity of the network in Fig.~\ref{fig_specialRigidExample}(b) has also been discussed in \cite[Fig.~5]{TronRigidComponent}.


There are several directions for future research. The first is the minimum number of edges required to ensure generic bearing rigidity of a graph. This paper has showed that this minimum number is $2n-3$ for the two-dimensional case, but it is still unclear for three or higher dimensions. Other problems such as merging multiple bearing rigid networks and optimal design of edges to maximize the bearing rigidity degree also deserve more research attention.

\bibliography{myOwnPub,zsyReferenceAll} 
\bibliographystyle{ieeetr}

\end{document}